\documentclass[12pt]{article}
\usepackage{amsmath}
\usepackage{amssymb}
\usepackage{amsfonts}
\usepackage{latexsym}
\usepackage{color}
\usepackage{graphicx}

\catcode `\@=11 \@addtoreset{equation}{section}

\catcode `\@=12

\newcommand{\be}{\begin{equation}}
\newcommand{\en}{\end{equation}}
\newcommand{\bea}{\begin{eqnarray}}
\newcommand{\ena}{\end{eqnarray}}
\newcommand{\beano}{\begin{eqnarray*}}
\newcommand{\enano}{\end{eqnarray*}}
\newcommand{\bee}{\begin{enumerate}}
\newcommand{\ene}{\end{enumerate}}

\newcommand{\Hil}{{\cal H}}

\newcommand{\F}{{\cal F}}
\newcommand{\Lc}{{\cal L}}
\newcommand{\Sc}{{\cal S}}
\newcommand{\D}{{\cal D}}
\newcommand{\K}{{\cal K}}

\newcommand{\M}{{\cal M}}

\newcommand{\ltwo}{{\Lc^2(\mathbb{R})}}
\newcommand{\scr}{{\Sc(\mathbb{R})}}
\newcommand{\1}{1 \!\! 1}
\newtheorem{thm}{Theorem}
\newtheorem{cor}[thm]{Corollary}
\newtheorem{lemma}[thm]{Lemma}
\newtheorem{prop}[thm]{Proposition}
\newtheorem{rem}[thm]{Remark}
\newtheorem{defn}[thm]{Definition}

\newtheorem{example}[thm]{Example}
\newenvironment{proof}{\noindent {\bf Proof:}}{\hfill$\Box$}
\begin{document}
\thispagestyle{empty}
\vspace*{1cm}

\begin{center}
{\Large \bf Unbounded Hamiltonians generated by Parseval frames} 

\vspace{4mm}

{\large F. Bagarello${}^{1,2}$, S. Ku\.{z}el${}^3$}\\
\vspace*{1cm}

\normalsize
${}^1$Dipartimento di Ingegneria, Universit\`a di Palermo, I-90128  Palermo, Italy\\

\vspace*{.5cm}
${}^2$ INFN, Sezione di Catania, Italy\\

\vspace*{.5cm}
${}^3$ AGH University, Krak\'ow, Poland.\\

\end{center}

\vspace*{0.5cm}

\begin{abstract}
	 {In \cite{BK} Parseval frames were used to define bounded Hamiltonians, both in finite and in infinite dimesional Hilbert spaces. Here we continue this analysis, with a particular focus on the discrete spectrum of Hamiltonian operators defined as a weighted infinite sum of rank one operators defined by some Parseval frame living in an infinite dimensional Hilbert space. The main difference with \cite{BK} is that, here, the operators we consider are mostly unbounded. This is an useful upgrade with respect to our previous results, since physically meaningful Hamiltonians are indeed often unbounded. However, due to the fact that  frames (in general) are not bases, the definition of an Hamiltonian is not so easy, and part of our results goes in this direction. Also, we discuss the eigenvalues of the Hamiltonians, and we discuss some physical applications of our framework. }
\end{abstract}

\vspace{2cm}


\vfill

\newpage

\section{Introduction}\label{sect1}

In quantum mechanics a central object is the Hamiltonian $H$ 
of the physical system $\Sc$ one is interested in. $H$ is, for closed and conservative systems, the energy of $\Sc$. 
Quite often, the first step to analyze $\Sc$ consists in finding the eigenvectors of $H$, 
$e_j$: $H\,e_j=E_je_j$, $j=1,2,3,\ldots$. If $H$ is self-adjoint, then each $E_j\in\mathbb{R}$ and 
eigenvectors related to different eigenvalues are orthogonal: $\langle e_j,e_k\rangle=0$, if $E_j\neq E_k$. 
Most of the times $\F_e=\{e_j\}$ is an orthonormal basis (ONB) for $\Hil$, the Hilbert space where $\Sc$ is defined, and the Hamiltonian is written 
\be
H=\sum_jE_j\langle e_j,\cdot\rangle\,e_j=\sum_jE_jP_j,
\label{11}\en
where we have introduced $P_j$ acting on a generic vector $f$ of $\Hil$ as follows: $P_jf=\langle e_j,f\rangle\,e_j$. Since $P_j=P_j^*=P_j^2$, $H$ can be seen as a weighted
sum of orthogonal projectors, with weights given by its eigenvalues $E_j$. If is clear that the domain of $H$, $\mathcal{D}(H)$, is not necessarily all of $\Hil$. 
In fact we have $\mathcal{D}(H)=\{f\in\Hil:\,\, \sum_j{E_j^2}|\langle e_j,f\rangle|^2<\infty\}$, which is surely dense in $\Hil$ (it contains the linear span of the $e_j$'s), but not necessarily coincident with $\Hil$. 

In the past 30 years it has become clearer and clearer that the Hamiltonian of a system does not really have to be self-adjoint, \cite{benboet}. Since this seminal paper, an always increasing number of physicists and mathematicians started to consider this possibility, 
where the reality of the eigenvalues of a manifestly non self-adjoint Hamiltonian $H$ is due to some physical symmetry rather than on the mathematical requirement that $H=H^*$.
 However, even if the eigenvalues of $H$ are real, the eigenvectors, $\varphi_j$ are no longer mutually orthogonal, in general. Still, if the set $\F_\varphi=\{\varphi_j\}$ is a basis for $\Hil$, a second (uniquely determined) basis of the Hilbert space also exists, $\F_\psi=\{\psi_j\}$, such that $\langle\varphi_j,\psi_k\rangle=\delta_{j,k}$, 
 and $H^*\psi_k=\overline{E_k}\,\psi_k$. In this case we can write $H$ and $H^*$ as follows:
$$
H=\sum_jE_j\langle \psi_j,\cdot\rangle\,\varphi_j=\sum_jE_jQ_j, \qquad H^*=\sum_j\overline{E_j}\langle \varphi_j,\cdot\rangle\,\psi_j=\sum_j\overline{E_j}\,Q_j^*,
$$
where $Q_jf=\langle\psi_j,f\rangle\,\varphi_j$ is a projection operator, but 
it is not orthogonal if $\varphi_j\neq\psi_j$. Operators of this kind have been studied in the past, mainly from a mathematical point of view. 
We refer to \cite{bit2014} for some results, mainly for the case in which $\F_\varphi$ and $\F_\psi$ are Riesz bases. 
Later, generalizations of this situation have been considered, \cite{bagbell,hinoue,hinouetaka}. In all these extensions, biorthogonality of the sets of vectors used to 
define some specific Hamiltonian was required. In \cite{BK}, this assumption was removed, in our knowledge, for the first time: we used frames, and in particular Parseval frames 
(PF), rather than bases. Hence the existence of a biorhogonal basis is not guaranteed. 
Our interest was mainly mathematical, but was also based on a very simple physical remark:   
in the analysis of a concrete physical situation it may happen that \emph{not all} {vectors of $\Hil$} are {\em relevant} in the analysis of $\Sc$. 
For instance, when the energy of $\Sc$ cannot {\em really} increase too much, or when $\Sc$ is localized in a bounded region, or still when the value of the momentum 
of $\Sc$ cannot be too large. In all these cases, but not only, it is reasonable to consider a {\em physical vector space}, $\Hil_{ph}$, as the subset of the
{\em mathematical} Hilbert space $\Hil$ on which $\Sc$ is originally defined. 
 $\Hil_{ph}$ can be constructed as the projection of $\Hil$, via some suitable orthogonal projector operator 
$P$ i.e.,   $\Hil_{ph}=P\Hil$. The restriction of $H$ defined by \eqref{11} onto $\Hil_{ph}$ gives rise to the new Hamiltonian $ H_{ph}=PHP$ (physical part
of $H$) acting in $\Hil_{ph}$
\begin{equation}\label{K2}
	H_{ph}f=\sum{E}_j \langle \varphi_j, f\rangle\,\varphi_j, 
\end{equation} 
where,  the set of vectors $\F_\varphi=\{\varphi_j=Pe_j\}$  loses the property of being an ONB of 
$\Hil_{ph}$ and, instead, it turns out to be a  Parseval frame of $\Hil_{ph}$.

In this paper we explore further this situation, extending what we have found in \cite{BK} to unbounded operators, which are often more relevant in the analysis of concrete physical systems, \cite{br1,hall}.

The paper is divided in two parts: in Section \ref{sectII} we describe our mathematical results, while Section \ref{sect3} contains some detailed examples. More in details, after some preliminaries in Section \ref{sII.1}, in Section \ref{sII.2} we introduce the Hamiltonian $H_{{\varphi}, \sf E}=\sum_{j\in\mathbb{J}}E_j \langle \varphi_j, \cdot\rangle\,\varphi_j$, for a PF $\{\varphi_j\}$, and we study its domain, and we give conditions for $H_{{\varphi}, \sf E}$ to be, or not, bounded and self-adjoint. Its spectrum is then analyzed 
in Section \ref{sII.4}.

In Section \ref{s3.1} we use a PF first introduced by Casazza and Christensen to define a particular $H_{{\varphi}, \sf E}$, and in particular to study its eigenvalues and eigenvectors. We also show how to introduce two different sets of ladder operators, $a_n$ and $V_n$. In particular, $a_n$ and its adjoint $a_n^*$ obey a truncated version of the canonical commutation relations which was first used, in our knowledge, in \cite{bagchi}, and later extended in \cite{bag2018}. These truncated bosonic operators allow to interpolate between fermions and bosons, going from $n=2$ to $n\rightarrow\infty$. We will also comment that the operators $V_n$ may have an interesting role in signal analysis.

Section \ref{s3.2} contains a first part, where we propose a particular method to construct a PF out of an ONB, and then we apply this method to the Hamiltonian of a single electron in a strong magnetic field. Its eigenstates produce the so-called  Landau levels, which have infinite degeneracy. This Hamiltonian is at the basis of the analysis of the quantum Hall effect. In particular, we will show that our approach could be seen as generating two different lattices in each Landau levels.

\section{Hamiltonians generated by Parseval frames}\label{sectII}

\subsection{Parseval frames.}\label{sII.1}
Here all necessary facts about Parseval frames   
are presented in a form convenient for our exposition. More details can be found in \cite{chri, heil}.

Let $\K$ be a complex infinite-dimensional separable Hilbert space with scalar product $\langle \cdot, \cdot \rangle$ 
linear in the second argument.  Denote by $\mathbb{J}$  a generic countable index set such as
$\mathbb{Z}$, $\mathbb{N}$, $\mathbb{N}\cup\{0\}$, etc. 

A \emph{Parseval  frame}  (PF for short) is a family of vectors $\F_\varphi=\{\varphi_j, \,j\in \mathbb{J}\}$  in $\K$  which satisfies
\begin{equation}\label{uman1}
\sum_{j\in\mathbb{J}}{|\langle \varphi_j, f\rangle|^2}=||f||^2, \quad f\in \K.
\end{equation}
It follows from \eqref{uman1} that $\|\varphi_j\|\leq{1}$ for $j\in\mathbb{J}$.

According to the Naimark dilation theorem, each PF in $\K$ can be extended to an orthonormal basis of a wider subspace $\Hil$. 
\begin{thm}[\cite{hanlar}]\label{Naimark}
Let $\F_\varphi=\{\varphi_j, \,j\in \mathbb{J}\}$ be a PF in a Hilbert space $\K$. Then there exists a complementary Hilbert space $\M$ and a PF
$\F_\psi=\{\psi_j, \,j\in \mathbb{J}\}$ in $\M$ such that
\begin{equation}\label{K3b}
\F_h=\{h_j=\varphi_j\oplus\psi_j, \ \  j\in \mathbb{J}\}
\end{equation} 
is an ONB for $\Hil=\K\oplus\M$. 
\end{thm}

It is also easy to prove, with a direct computation, that taken an ONB $\{e_n\}$ of $\Hil$, 
and an orthogonal projector $P$, the set $\{Pe_n\}$ is a PF in $\Hil_{ph}=P\Hil$, as already claimed in the Introduction.

The \emph{excess} $e[\F_\varphi]$ of a PF of $\F_\varphi=\{\varphi_j, \,j\in \mathbb{J}\}$  is the greatest integer $n$ such that  $n$ elements  
can be deleted from the frame $\F_\varphi$ and still leave a complete set, or $\infty$ if there is no upper bound to the number of elements that can be removed.
It follows from \cite[Lemma 4.1]{BCCL} and Theorem \ref{Naimark} that  $e[\F_\varphi]$ coincides with $\dim\M$,  where
$\M$ is  the complementary Hilbert space in Theorem \ref{Naimark}. 
 The zero excess means that  $\F_\varphi$ 
is an ONB  of $\K$.  The finite excess $e[\F_\varphi]$ means that the index set $\mathbb{J}$ can be decomposed
 $\mathbb{J}=\mathbb{J}_0\cup\mathbb{J}_1$ in such a way that 
 $\F_\varphi^0=\{\varphi_j, j\in\mathbb{J}_0\}$ is a Riesz basis in $\K$ and $\mathbb{J}_1$ is a finite set \cite{CC98}. 

Each PF  $\F_\varphi$ determines an isometric operator $\theta_\varphi : \K \to \ell_2(\mathbb{J})$:
\begin{equation}\label{K31}
\theta_\varphi{f}=\{\langle \varphi_j, f\rangle\}_{j\in\mathbb{J}}, \qquad f\in\K,
\end{equation}
which is called an \emph{analysis operator} associated with $\F_\varphi$. 

The adjoint operator $\theta_\varphi^* : \ell_2(\mathbb{J}) \to \K$ of  $\theta_\varphi$ is called a \emph{synthesis operator} and it 
acts as follows 
\begin{equation}\label{K12b}
\theta_\varphi^*\{c_j\}=\sum_{j\in\mathbb{J}}c_j\varphi_j, \qquad  \{c_j\}\in\ell_2(\mathbb{J}).
\end{equation}

Let $\theta_\varphi$ and $\theta_\psi$ be analysis operators associated with  PF's  $\F_\varphi$ and  $\F_\psi$  from  Theorem \ref{Naimark}. Then 
\begin{equation}\label{abba1}
\ell_2(\mathbb{J})={\mathcal R}(\theta_\varphi)\oplus{\mathcal R}(\theta_\psi),
\end{equation}
where ${\mathcal R}(\theta_\varphi)$ and ${\mathcal R}(\theta_\psi)$ are 
the image sets of the operators $\theta_\varphi$ and $\theta_\psi$, respectively.

By virtue of  \eqref{K12b} and \eqref{abba1},  for  $\{c_j\}\in\ell_2(\mathbb{J})$, the following relation holds
\begin{equation}\label{K13}
\sum_{j\in\mathbb{J}}c_j\varphi_j=0 \quad \iff \quad  \{c_j\}\in\ker\theta_\varphi^*=\ell_2(\mathbb{J})\ominus{\mathcal R}(\theta_\varphi)={\mathcal R}(\theta_\psi).
\end{equation}

\subsection{Operators $H_{{\varphi}\sf E}^{min}$ and $H_{{\varphi}\sf E}^{max}$.}\label{sII.2}
For a given  PF  $\F_{\varphi}=\{\varphi_j,  j\in\mathbb{J}\}$  and a sequence of real quantities  ${\sf E}=\{E_j,  j\in\mathbb{J}\}$, one can introduce a linear operator
\begin{equation}\label{K5}
H_{{\varphi}, \sf E}=\sum_{j\in\mathbb{J}}E_j \langle \varphi_j, \cdot\rangle\,\varphi_j
\end{equation}
in a Hilbert space $\K$. If $e[\F_{\varphi}]=0$  (i.e., $\F_{\varphi}$ is an ONB of $\K$), the quantities $E_j$ in \eqref{K5} turns out to be eigenvalues of $H_{{\varphi}\sf E}$. 
If $\F_\varphi$ is a PF (but not an ONB), this fact does not hold in general (see Lemma \ref{palermo1} below). 

The operator $H_{{\varphi}\sf E}$ may be unbounded and its domain of definition should be specified.  
There are two natural domains {for $H_{{\varphi}\sf E}$ typically used in the literature:}
\begin{equation}\label{uman6}
\begin{array}{l}
  \mathcal{D}_{min}=\{f\in\K \ : \ \sum_{j\in\mathbb{J}}E_j^2 |\langle \varphi_j,  f\rangle|^2<\infty\}, \vspace{3mm} \\
\mathcal{D}_{max}=\{f\in\K \ : \ \sum_{j\in\mathbb{J}}E_j \langle \varphi_j,  f\rangle\,\varphi_j \quad \mbox{converges unconditionally in} \quad \K\}.
\end{array}
\end{equation}
{It follows from \cite[Theorem 7.2 (b)]{heil}  that}
\begin{equation}\label{uman2}
 \mathcal{D}_{min}\subseteq\mathcal{D}_{max}.
 \end{equation}
We observe that it may easily happen in \eqref{uman2} that  $\mathcal{D}_{min}$ is indeed a proper
subspaces of $\mathcal{D}_{max}$.


\begin{example}
Set $\mathbb{J}=\mathbb{Z}\setminus\{0\}$ and consider
 the quantities $E_n=n^2$, $E_{-n}=0$ $(n\in\mathbb{N})$ and an ONB $\{{e}_n\}_{n\in\mathbb{N}}$ of $\K$. 
  In view of \cite[Example 8.35]{heil}, the set of vectors $\F_\varphi=\{\varphi_j,  \ j\in\mathbb{J}\}$, where
 $$
 \varphi_n=\frac{1}{n}{e}_n,  \qquad   \varphi_{-n}=\sqrt{1-\frac{1}{n^2}}{e}_n 
 $$
 is a PF.  In this case,  the series 
 $$
 \sum_{j\in\mathbb{J}}E_j \langle \varphi_j,  f\rangle\,\varphi_j=\sum_{n\in\mathbb{N}}\langle {e}_n,  f\rangle\, {e}_n
 $$
 converges unconditionally for all $f\in\K$. Hence,  $\mathcal{D}_{max}=\K$. 
On the other hand,
$\mathcal{D}_{min}=\{f\in\K: \sum_{n\in\mathbb{N}}n^2 |\langle {e}_n,  f\rangle|^2<\infty\}$ is a subset of $\K$. 
Therefore, $\mathcal{D}_{min}\subset\mathcal{D}_{max}$.
\end{example}

The specificity of frames gives rise to the following curious  fact. 
 
\begin{prop}\label{prop1}
For arbitrary unbounded set of quantities  ${\sf E}$  there are uncountably many PF's $\F_\varphi$  such that
$\mathcal{D}_{min}$  is trivial, i.e. $\mathcal{D}_{min}=\{0\}$.
\end{prop}
\begin{proof}
Without loss of generality, one can assume that $\mathcal{H}=\ell^2(\mathbb{J})$, where $\mathbb{J}$ is the countable set of indices.
Let  $\mathcal{E}$ be an operator of multiplication by the set ${\sf E}=\{E_j,  j\in\mathbb{J}\}$  in $\ell^2(\mathbb{J})$:
 \begin{equation}\label{K201}
 \mathcal{E}\{c_j\}=\{E_jc_j\}, \qquad  \mathcal{D}(\mathcal{E})=\{\{c_j\}\in\ell^2(\mathbb{J}) :  \{E_jc_j\}\in\ell^2(\mathbb{J}) \}.
 \end{equation}
 Since the set ${\sf E}$ is unbounded (i.e. $\sup_{j\in\mathbb{J}}\{|E_j|\}=\infty$)
the self-adjoint operator $\mathcal{E}$ is unbounded in $\ell^2(\mathbb{J})$ and, by the extended version \cite[Theorem 3.19]{AZ} 
of the Schm\"{u}dgen theorem \cite[Theorem 5.1]{S1}, there are uncountably many
infinite-dimensional subspaces $\K$ of  $\ell^2(\mathbb{J})$ such that 
\begin{equation}\label{palermo2023}
\mathcal{D}(\mathcal{E})\cap\K=\{0\}.
\end{equation}
For given $\K$ satisfying \eqref{palermo2023}  we denote by $P$ the orthogonal projection operator in  $\ell_2(\mathbb{J})$
on $\K$ and consider the canonical ONB  $\{e_j, \ j\in\mathbb{J}\}$ of 
$\ell_2(\mathbb{J})$.
Then $\F_\varphi=\{\varphi_j=Pe_j, \ j\in\mathbb{J}\}$ is
a PF in $\mathcal{K}$. The associated analysis operator $\theta_\varphi$ (see \eqref{K31}) 
maps $\K$ into $\ell^2(\mathbb{J})$  and 
$$
\theta_\varphi{f}=\{\langle \varphi_j, f \rangle\}_{j\in\mathbb{J}}=\{\langle e_j, f \rangle\}_{j\in\mathbb{J}}=\{c_j\}=f, \qquad  f=\{c_j\}\in\K.
$$
Therefore, $\mathcal{R}(\theta_\varphi)=\K$ and, in view of \eqref{K201}, \eqref{palermo2023},  $\{E_jc_j\}\not\in\ell^2(\mathbb{J})$ 
 for non-zero $f=\{c_j\}\in\mathcal{K}$. This means that
$$
\sum_{j\in\mathbb{J}}E_j^2 |\langle \varphi_j,  f\rangle|^2=\sum_{j\in\mathbb{J}}E_j^2 |\langle e_j,  f\rangle|^2=\sum_{j\in\mathbb{J}}E_j^2 |c_j|^2=\infty, \quad f\in\K, \ f\not=0
$$
and  $\mathcal{D}_{min}=\{0\}$.
\end{proof}

Of course, the choice of PF's $\F_\varphi$ in Proposition \ref{prop1} should be very specific. 
{Considering special classes of ${\sf E}$ and $\F_\varphi$ 
one can guarantee that  $\mathcal{D}_{min}$ coincides with $\mathcal{D}_{max}$ and  is a dense set in $\K$.} Few simple sufficient conditions are given below. 

\begin{prop}\label{gu11bb} 
The following assertions are true:
\begin{enumerate}
\item[(i)] if $\sup_{j\in\mathbb{J}}\{|E_j|\}<\infty$, then  
$\mathcal{D}_{min}=\mathcal{D}_{max}=\K$;
\item[(ii)] if  the index set $\mathbb{J}$ of  $\F_\varphi$ can be decomposed $\mathbb{J}=\mathbb{J}_0\cup\mathbb{J}_1$ in such a way that 
$\{\varphi_j, j\in\mathbb{J}_0\}$ is a Riesz basis in $\K$ and $\sup_{j\in\mathbb{J}_1}\{|E_j|\}<\infty$,  then $\mathcal{D}_{min}$ coincides with $\mathcal{D}_{max}$
and is a dense set in $\K$.
\end{enumerate} 
\end{prop}
\begin{proof} $(i)$
Denote $\alpha=\sup_{j\in\mathbb{J}}\{|E_j|\}$. Then
\begin{equation}\label{k24}
\sum_{j\in\mathbb{J}}E_j^2 |\langle \varphi_j,  f\rangle|^2\leq\alpha^2\sum_{j\in\mathbb{J}}|\langle \varphi_j,  f\rangle|^2=\alpha^2\|f\|^2, \qquad f\in\K.
\end{equation}
Therefore, $\mathcal{D}_{min}=\K$. By virtue of \eqref{uman2},  $\mathcal{D}_{min}=\mathcal{D}_{max}=\K$.

$(ii)$ Similarly to \eqref{k24},  $\sum_{j\in\mathbb{J}_1}E_j^2 |\langle \varphi_j,  f\rangle|^2\leq\alpha^2\sum_{j\in\mathbb{J}_1}|\langle \varphi_j,  f\rangle|^2\leq\alpha^2\|f\|^2$, 
where $\alpha=\sup_{j\in\mathbb{J}_1}\{|E_j|\}$ and  $f\in\K.$ This means that the sets  $\mathcal{D}_{min}$ and $\mathcal{D}_{max}$ defined by \eqref{uman6}
do not change if we consider the smaller set $\mathbb{J}_0=\mathbb{J}\setminus\mathbb{J}_1$ instead of $\mathbb{J}$.

Each Riesz basis $\{\varphi_j,  j\in\mathbb{J}_0\}$ is norm-bounded below (i.e.,  $\inf_{j\in\mathbb{J}}\|\varphi_j\|>0$). By
\cite[Theorem 8.36]{heil} this means that  $\mathcal{D}_{min}=\mathcal{D}_{max}$.  
Let $\{\psi_j,  j\in\mathbb{J}_0\}$ be the biorthogonal Riesz basis for $\{\varphi_j,  j\in\mathbb{J}_0\}$. By virtue of \eqref{uman6}, 
$\psi_j\in\mathcal{D}_{min}$ for all $j\in\mathbb{J}_0$. Hence, $\mathcal{D}_{min}$ 
 is a dense set in $\K$.  
\end{proof}

In what follows we consider a slightly more general case assuming that \textit{$\mathcal{D}_{min}$ is dense in $\K$ 
and $\mathcal{D}_{min}\subseteq\mathcal{D}_{max}$, so that $\mathcal{D}_{max}$ is dense in $\K$ a fortiori.}

Equipping \eqref{K5} with domains $ \mathcal{D}_{min}$ and $\mathcal{D}_{max}$
we define the following operators in $\K$:
\begin{equation}\label{K22}
\begin{array}{l}
  H_{{\varphi}\sf E}^{min}f=\sum_{j\in\mathbb{J}}E_j \langle \varphi_j, f\rangle\,\varphi_j, \qquad f\in\mathcal{D}(H_{{\varphi}\sf E}^{min})=\mathcal{D}_{min},  \vspace{3mm} \\
H_{{\varphi}\sf E}^{max}f=\sum_{j\in\mathbb{J}}E_j \langle \varphi_j, f\rangle\,\varphi_j, \qquad f\in\mathcal{D}(H_{{\varphi}\sf E}^{max})=\mathcal{D}_{max}.
\end{array}
\end{equation}

The operator $H_{{\varphi}\sf E}^{min}$  admits  a simple interpretation with the use
of Naimark dilation  theorem.  Namely, in the Hilbert space $\Hil=\K\oplus\M$, we consider
the ONB $\{h_j, \  j\in \mathbb{J}\}$ defined by \eqref{K3b}
and define a self-adjoint Hamiltonian
\begin{equation}\label{K5b}
H_{h{\sf E}}=\sum_{j\in\mathbb{J}}E_j \langle h_j, \cdot \rangle\,h_j,  \qquad  \mathcal{D}(H_{h{\sf E}})=\{f\in\Hil \ : \ \sum_{j\in\mathbb{J}}E_j^2 |\langle h_j, f\rangle|^2<\infty\}.
\end{equation} 
By the construction,  the set of eigenvalues of $H_{h{\sf E}}$ coincides with ${\sf E}=\{E_j\}$ and the corresponding eigenfunctions are $\{h_j\}$. 
Comparing \eqref{K22} and \eqref{K5b} we arrive at the conclusion that
\begin{equation}\label{important} 
H_{{\varphi}\sf E}^{min}f=P_{\K}H_{h{\sf E}}f, \qquad f\in\mathcal{D}(H_{{\varphi}\sf E}^{min})=\mathcal{D}_{min}=\mathcal{D}(H_{h{\sf E}})\cap{\K}.
\end{equation}
Therefore, the operator $H_{{\varphi}\sf E}^{min}$ may be interpreted as
the restriction of the Hamiltonian $H_{h{\sf E}}$ acting in $\Hil=\K\oplus\M$ to the physical space $\Hil_{ph}=\K$, see \eqref{K2}.

It follows from  \eqref{uman2} that 
\begin{equation}\label{uman63}
H_{{\varphi}\sf E}^{min}\subseteq{H_{{\varphi}\sf E}^{max}}.
\end{equation}
The densely defined operator   $H_{{\varphi}\sf E}^{max}$ is symmetric in $\K$  since
$$
\langle H_{{\varphi}\sf E}^{max}f, f\rangle=\sum_{j\in\mathbb{J}}E_j |\langle \varphi_j, f\rangle|^2,  \qquad f\in\mathcal{D}(H_{{\varphi}\sf E}^{max}) 
$$
is real-valued. The same is true for $H_{{\varphi}\sf E}^{min}$ since \eqref{uman63} holds.
\begin{cor}\label{dd1}
If $H_{{\varphi}\sf E}^{min}$ is self-adjoint in $\K$, then $H_{{\varphi}\sf E}^{max}$ is also self-adjoint and $H_{{\varphi}\sf E}^{min}=H_{{\varphi}\sf E}^{max}$. 
\end{cor}
\begin{proof}
It follows immediately from the relation
$$
H_{{\varphi}\sf E}^{min}\subseteq{H_{{\varphi}\sf E}^{max}}\subseteq(H_{{\varphi}\sf E}^{max})^*\subseteq(H_{{\varphi}\sf E}^{min})^*=H_{{\varphi}\sf E}^{min}.
$$
\end{proof}

 Using Proposition \ref{gu11bb}, it is easy to derive sufficient conditions for self-adjointness of
$H_{{\varphi}\sf E}^{min}$. 

\begin{cor}\label{uman7} 
The following assertions are true:
\begin{enumerate}
\item[(i)] if $\sup_{j\in\mathbb{J}}\{|E_j|\}<\infty$, then $H_{{\varphi}\sf E}^{min}$ is a bounded self-adjoint operator in $\K$; 
\item[(ii)] if $\sup_{j\in\mathbb{J}}\{|E_j|\}=\infty$ and the index set $\mathbb{J}$ of  $\F_\varphi$ can be decomposed $\mathbb{J}=\mathbb{J}_0\cup\mathbb{J}_1$ in such a way that 
$\{\varphi_j, j\in\mathbb{J}_0\}$ is a Riesz basis in $\K$ and $\sup_{j\in\mathbb{J}_1}\{|E_j|\}<\infty$,  then
$H_{{\varphi}\sf E}^{min}$ is an unbounded self-adjoint operator in $\K$.
\end{enumerate} 
\end{cor}
\begin{proof} $(i)$. By Proposition \ref{gu11bb} and \eqref{K22}, the symmetric operator $H_{{\varphi}\sf E}^{min}$ is defined on the whole space $\K$. Hence,
 $H_{{\varphi}\sf E}^{min}$ is a bounded self-adjoint operator. 

$(ii)$. By employing Proposition \ref{gu11bb} once again,  one gets that 
$H_{\varphi\sf E}^{min}$ is a densely defined operator in $\K$.
 For $f\in\mathcal{D}(H_{{\varphi}\sf E})$ we decompose
 \begin{equation}\label{HH1} 
  H_{{\varphi}\sf E}f=\sum_{j\in\mathbb{J}_0}E_j \langle \varphi_j, f\rangle\,\varphi_j+\sum_{j\in\mathbb{J}_1}E_j \langle \varphi_j, f\rangle\,\varphi_j, 
  \end{equation} 
where $\F_\varphi^0=\{\varphi_j, \ j\in\mathbb{J}_0\}$ is a Riesz basis in $\K$ and
 $\sup_{j\in\mathbb{J}_1}\{|E_j|\}<\infty$.

Denote by $S_0=\sum_{j\in\mathbb{J}_0}\langle \varphi_j, \cdot\rangle\,\varphi_j$ the frame operator of the  Riesz basis  $\F_\varphi^0$. Then  \cite[Theorem 6.1.1]{chri}
  $$
  \varphi_j=S_0^{1/2}e_j, \quad j\in\mathbb{J}_0,
  $$
  where $\{e_j, \ j\in\mathbb{J}_0\}$ is an ONB  of $\K$. 
Therefore, the first operator in  \eqref{HH1},  $$
A_0=\sum_{j\in\mathbb{J}_0}E_j \langle \varphi_j, \cdot\rangle\,\varphi_j, \quad \mathcal{D}(A_0)=\{f\in\K \ : \ \sum_{j\in\mathbb{J}_0}E_j^2 |\langle \varphi_j,  f\rangle|^2<\infty\}
$$
 can be rewritten as follows
 \begin{equation}\label{bbb1}
 A_0=S_0^{1/2}H_{e{\sf E}_0}S_0^{1/2}, \qquad H_{e{\sf E}_0}=\sum_{j\in\mathbb{J}_0}E_j \langle e_j, \cdot\rangle\,e_j, \quad {\sf E}_0=\{E_j, j\in\mathbb{J}_0\}.
 \end{equation}
By the construction, the operator $H_{e{\sf E}_0}$ with the domain 
 $$\mathcal{D}(H_{e{\sf E}_0})=\{f\in\K \ : \ \sum_{j\in\mathbb{J}_0}E_j^2 |\langle e_j,  f\rangle|^2<\infty\}$$ is self-adjoint in $\K$.
In view of \eqref{bbb1}, $\mathcal{D}(H_{e{\sf E}_0})=S_0^{{1}/2}\mathcal{D}(A_0)$ and 
\begin{equation}\label{bbb2}
\langle A_0f, g\rangle=\langle H_{e{\sf E}_0}S_0^{1/2}f, S_0^{1/2}g\rangle, \qquad f,g\in\mathcal{D}(A_0).
\end{equation}
Relation \eqref{bbb2} means that  $A_0$ is self-adjoint. Moreover, as follows from the proof of Proposition \ref{gu11bb},  
$\mathcal{D}(A_0)=\{f\in\K \ : \ \sum_{j\in\mathbb{J}_0}E_j^2 |\langle \varphi_j,  f\rangle|^2<\infty\}=\mathcal{D}_{min}=\mathcal{D}(H_{\varphi{\sf E}}^{min})$.

On the other hand,    
$$
\sum_{j\in\mathbb{J}_1}E_j^2 |\langle \varphi_j,  f\rangle|^2\leq\alpha^2\sum_{j\in\mathbb{J}_1}|\langle \varphi_j,  f\rangle|^2<\alpha^2\|f\|^2, \quad f\in\K, \quad  \alpha=\sup_{j\in\mathbb{J}_1}\{|E_j|\}.
$$
This means (see, e.g., \cite[Theorem 7.2]{heil}) that  the second operator in \eqref{HH1} 
  $$
  A_1=\sum_{j\in\mathbb{J}_1}E_j \langle \varphi_j, \cdot\rangle\,\varphi_j
  $$ 
is defined on $\K$ and it is  symmetric 
(since  $\langle A_1f, f \rangle=\sum_{j\in\mathbb{J}_1}E_j |\langle \varphi_j, f\rangle|^2$ is real-valued for $f\in\K$). Hence, $A_1$ is a bounded self-adjoint operator in $\K$.
 This means that the operator $H_{{\varphi}\sf E}=A_0+A_1$ with the domain $\mathcal{D}(A_0)=\mathcal{D}(H_{{\varphi}\sf E}^{min})$
is  self-adjoint. 
\end{proof}
\begin{rem}
Similar  results (in the case where $\F_{\varphi}$ involves a Riesz basis) can be obtained by the perturbation theory methods if the operator $A_1$ is sufficiently small with respect to the self-adjoint operator $A_0$ (see, e.g., \cite[X.2]{ReedSimonII}).
\end{rem}

\subsection{Spectrum of $H_{{\varphi}\sf E}$.}\label{sII.4}

In what follows, \emph{we suppose that $H_{{\varphi}\sf E}^{min}$ is a self-adjoint operator}
in $\K$. In this case, in view of {Corollary \ref{dd1}, $H_{{\varphi}\sf E}:=H_{{\varphi}\sf E}^{min}=H_{{\varphi}\sf E}^{max}$ is a self-adjoint} 
operator.
 
We consider  $H_{{\varphi}\sf E}$  as 
\emph{a Hamiltonian generated by the pair $(\F_\varphi, {\sf E})$ of a PF $\F_\varphi$ and a set of real numbers ${\sf E}$}.
As was mentioned above, the Hamiltonian $H_{{\varphi}\sf E}$ is  the restriction to the physical space $\K$ of the Hamiltonian $H_{h{\sf E}}$ (see \eqref{K5b})  acting in the wider Hilbert space $\Hil=\K\oplus\M$.
By  virtue of \eqref{K5b}, the point spectrum of $H_{h{\sf E}}$ coincides with the set ${\sf E}$ and $H_{h{\sf E}}h_n=E_nh_n$ \ $(n\in\mathbb{J})$.  On the other hand: 

\begin{lemma}\label{palermo1} 
The following assertions are equivalent:
\begin{enumerate}
\item[(i)] the relation $H_{{\varphi}\sf E}\varphi_n=E_n\varphi_n$ holds for some $E_n\in{\sf E}$; 
\item[(ii)] the vectors of the PF $\F_\psi$ in Theorem \ref{Naimark}  satisfy the following conditions: 
\begin{equation}\label{uman1963}
\sum_{j\in\mathbb{J}}E_j^2|\langle\psi_j, \psi_n\rangle|^2<\infty, \qquad  \sum_{j\in\mathbb{J}}E_j\langle\psi_j, \psi_n\rangle\varphi_j=0
\end{equation}
\end{enumerate} 
\end{lemma}
\begin{proof} $(i)\to(ii)$.
 If $H_{{\varphi}\sf E}\varphi_n=E_n\varphi_n$, then $\varphi_n\in\mathcal{D}(H_{{\varphi}\sf E})\cap\mathcal{D}(H_{h\sf E})$ (see \eqref{important}).
This means that the vector $\psi_n$ corresponding to  $\varphi_n$ in \eqref{K3b} also belongs to $\mathcal{D}(H_{h\sf E})$. Therefore,
$\sum_{j\in\mathbb{J}}E_j^2|\langle\psi_j, \psi_n\rangle|^2=\sum_{j\in\mathbb{J}}E_j^2|\langle{h}_j, \psi_n\rangle|^2<\infty.$ 
Further, the relation $H_{{h}\sf E}h_n=E_nh_n$ means that $P_{\K}H_{{h}\sf E}h_n=E_j\varphi_n$. Taking \eqref{K3b} and \eqref{K5b} into account, one gets
\begin{equation}\label{fufu1}
0=P_{\K}H_{{h}\sf E}h_n-H_{{\varphi}\sf E}\varphi_n=\sum_{j\in\mathbb{J}}E_j\langle h_j, h_n\rangle\varphi_j-\sum_{j\in\mathbb{J}}E_j\langle \varphi_j, \varphi_n\rangle\varphi_j=\sum_{j\in\mathbb{J}}E_j\langle \psi_j, \psi_n\rangle\varphi_j
\end{equation}
that establishes \eqref{uman1963}.

$(ii)\to(i)$. The first part in \eqref{uman1963} and \eqref{K5b} mean that
$\psi_n\in\mathcal{D}(H_{h\sf E})$. Hence, $\varphi_n=h_n-\psi_n$ belongs to $\mathcal{D}(H_{{\varphi}\sf E})\cap\mathcal{D}(H_{h\sf E})$.
Reasoning similarly to \eqref{fufu1},  we obtain
$$
0=\sum_{j\in\mathbb{J}}E_j\langle \psi_j, \psi_n\rangle\varphi_j=P_{\K}H_{{h}\sf E}h_n-H_{{\varphi}\sf E}\varphi_n=E_nP_{\K}h_n-H_{{\varphi}\sf E}\varphi_n=E_n\varphi_n-H_{{\varphi}\sf E}\varphi_n
$$
that completes the proof.
\end{proof}

If $\|\varphi_n\|=1$, then $\psi_n=0$ and the conditions \eqref{uman1963} are clearly satisfied.  In such case, $E_n$ turns out to be an eigenvalue of $H_{{\varphi}\sf E}$.

Lemma \ref{palermo1} shows that  
the set $\sf E$ not always coincides with the point spectrum of $H_{{\varphi}\sf E}$. 
For this reason, we will say that ${\sf E}$ is a set of \emph{quasi-eigenvalues} of $H_{\varphi{\sf E}}$.

Since $H_{{\varphi}\sf E}$ is assumed to be self-adjoint, one may believe that  $H_{{\varphi}\sf E}$ can also be presented in a form similar to \eqref{K5b}:
\begin{equation}\label{wu2b}
H_{{\varphi}\sf E}=\sum_{j\in\mathbb{J}}\,E_j' \langle {e}_j', \cdot\rangle {e}_j', \qquad \mathcal{D}(H_{{\varphi}\sf E})=\{f\in\K \ : \ \sum_{j\in\mathbb{J}}E_j'^2| \langle {e}_j', f\rangle|^2<\infty\},
\end{equation}  
where $\F_{e'}=\{{e}_j'\in\K, \ j\in{\mathbb{J}}\}$ is an ONB of $\K$ and ${\sf E'}=\{{E}_j', \ j\in {\mathbb{J}}\}$ 
is a set of real numbers. The formula \eqref{wu2b} is more convenient for spectral analysis because it immediately gives the set 
${\sf E'}$ of eigenvalues of $H_{{\varphi}\sf E}$. This means, 
according what proposed in \cite{BK} for bounded operators, that $(\F_\varphi,{\sf E})$ is ${\sf E}-$connected to $\F_{e'}$. 
Indeed, under the assumptions here, 
$H_{{\varphi}\sf E}=\sum_{j\in\mathbb{J}}E_j\langle \varphi_j, \cdot\rangle\varphi_j=\sum_{j\in\mathbb{J}}\,E_j' \langle {e}_j', \cdot\rangle {e}_j'$.

We recall before formulating the next statement that the PF $\F_\varphi$ can be extended to an ONB $\F_h$ in  
$\Hil=\K\oplus\M$ by adding a complementary PF $\F_\psi$ of $\M$ (as in Theorem \ref{Naimark}).

\begin{thm}\label{kkk31}
Let $H_{{\varphi}\sf E}$  be a Hamiltonian generated by the pair $(\F_\varphi, {\sf E})$, {where ${\sf E}=\{E_j, j\in\mathbb{J}\}$ 
is a strictly increasing sequence $\ldots<E_j<{E_{j+1}}\ldots$}
 and let the operator 
\begin{equation}\label{palermo33}
B=\sum_{j\in\mathbb{J}}E_j \langle \varphi_j, \cdot\rangle\,\psi_j \ :  \ \K \to \M
\end{equation}
be bounded. {Then $H_{{\varphi}\sf E}$ has the form \eqref{wu2b} and its discrete spectrum 
$\sigma_{disc}(H_{{\varphi}\sf E})$ coincides with ${\sf E'}$.}
\end{thm}
\begin{proof} {We can assume, without loss of generality, that $\mathbb{J}=\mathbb{N}$. Since ${\sf E}=\{E_j, j\in\mathbb{N}\}$
is strictly increasing, there exists $\lambda=\lim_{j\to\infty}E_j$.  When $\lambda$ is less than infinity, 
${\sf E}$ is a bounded set, and the condition of boundedness of $B$ in \eqref{palermo33} is automatically fulfilled. Moreover
the self-adjoint operator
$$
\lambda{I}-H_{h{\sf E}}=\sum_{j\in\mathbb{N}}(\lambda-E_j)\langle h_j, \cdot\rangle\,h_j,
$$
where $H_{h{\sf E}}$ is defined by  \eqref{K5b},
is positive,  bounded and compact\footnote{compactness follows from \cite[problem 132]{Halmoshbook}  since $\lim_{j\to\infty}(\lambda-E_j)=0$.} in $\Hil$.
The same properties hold true for the operator $\lambda{I}-H_{\varphi{\sf E}}=P_{\K}(\lambda{I}-H_{h{\sf E}})P_\K$ acting in $\K$.
Therefore, there exists an ONB $\{{e}_j'\}_{j\in\mathbb{N}}$ of $\K$ formed by  eigenvectors ${e}_j'$:
\begin{equation}\label{K998}
(\lambda{I}-H_{\varphi{\sf E}}){e}_j'=\mu_j{e}_j', \qquad j\in\mathbb{N},
\end{equation}
corresponding to the decreasing sequence of eigenvalues $\mu_j$. The set $\{\mu_j\}_{j\in\mathbb{N}}$ constitutes the discrete spectrum of 
$\lambda{I}-H_{\varphi{\sf E}}$. 
It follows from \eqref{K998} that $H_{\varphi{\sf E}}{e}_j'=(\lambda-\mu_j){e}_j'$. Therefore, 
the bounded operator $H_{\varphi{\sf E}}$ can be defined by \eqref{wu2b} with $E_j'=\lambda-\mu_j$ and its discrete spectrum 
$\sigma_{disc}(H_{{\varphi}\sf E})$ coincides with ${\sf E'}=\{ E_j'=\lambda-\mu_j, \ j\in\mathbb{N}\}$.}

{Assume now that $\lambda=\lim_{j\to\infty}E_j=\infty$. Then the operators  $H_{h{\sf E}}$  and $H_{{\varphi}\sf E}$ are both semi-bounded from below, 
and they are interconnected through the  relation \eqref{important}.
 The operator $H_{h{\sf E}}$ has a compact resolvent (since its eigenvectors $\{h_j\}$ form an ONB of $\Hil$ and the corresponding eigenvalues $\{E_j\}$ are an
 increasing sequence tending to $\infty$).}
 In this case, by means of \cite[Theorem XIII.64]{ReedSimonIV}, the set
$$
Y_{b{H_{h{\sf E}}}}=\{f\in\mathcal{D}(H_{h{\sf E}}) : \|f\|\leq{1}, \ \|H_{h{\sf E}}f\|\leq{b}\},
$$
is compact in $\mathcal{H}$ for every $b\geq{0}$. Consider the similar set  associated with $H_{\varphi{\sf E}}$
$$
Y_{b{H_{\varphi{\sf E}}}}=\{f\in\mathcal{D}(H_{\varphi{\sf E}}) : \|f\|\leq{1}, \ \|H_{\varphi{\sf E}}f\|\leq{b}\}.
$$ 
In view of \eqref{important} and \eqref{palermo33},   for all $f\in{Y_{b{H_{\varphi{\sf E}}}}},$
 $$
 P_{\M}H_{h{\sf E}}f=P_{\M}\sum_{j\in\mathbb{N}}E_j \langle h_j, f \rangle\,h_j=\sum_{j\in\mathbb{N}}E_j \langle \varphi_j,  f \rangle\,\psi_j=Bf,  
 $$ where
$P_{\M}$ is the orthogonal projection operator on $\M$ in $\Hil$. Furthermore, 
\begin{equation}\label{uddu}
\|H_{h{\sf E}}f\|^2=\|P_{\K}H_{h{\sf E}}f\|^2+\|P_{\M}H_{h{\sf E}}f\|^2=\|H_{\varphi{\sf E}}f\|^2+\|Bf\|^2\leq{b^2+c^2}, \quad f\in{Y_{b{H_{\varphi{\sf E}}}}},
\end{equation}
where $c=\sup_{f\in{Y_{b{H_{\varphi{\sf E}}}}}}\|Bf\|\leq\|B\|<\infty$. Therefore,  
\begin{equation}\label{uman14}
Y_{b{H_{\varphi{\sf E}}}}\subset{Y_{\sqrt{b^2+c^2}{H_{h{\sf E}}}}}.
\end{equation}

Consider a convergent sequence $\{f_n\}$ in $\K$, where $f_n{\in}Y_{b{H_{\varphi{\sf E}}}}$ and denote $f=\lim{f_n}$.  Obviously,
$f\in\K$ and $\|f\|\leq{1}$.  Moreover, $f\in{Y_{\sqrt{b^2+c^2}{H_{h{\sf E}}}}}$ by virtue of \eqref{uman14} and the compactness of  
${Y_{\sqrt{b^2+c^2}{H_{h{\sf E}}}}}$. 
 This means that  $f\in\mathcal{D}(H_{h{\sf E}})\cap\K=\mathcal{D}(H_{\varphi{\sf E}})$. Assume that $\|H_{\varphi{\sf E}}f\|>b$, i.e. 
 $\|H_{\varphi{\sf E}}f\|^2=b^2+\varepsilon$ ($\varepsilon>0$). By virtue of \eqref{uddu}, $\|Bf\|^2\leq{c^2-\varepsilon}$ that contradicts to the definition of $c$.
 Hence, $\|H_{\varphi{\sf E}}f\|\leq{b}$ and $f\in{Y_{b{H_{\varphi{\sf E}}}}}$.
We verify that  $Y_{b{H_{\varphi{\sf E}}}}$ is a closed set in $\K$.  This fact, the compactness of $Y_{b{H_{h{\sf E}}}}$, and \eqref{uman14} 
mean that $Y_{b{H_{\varphi{\sf E}}}}$  is a compact set in $\K$ for each $b\geq{0}$. Applying again 
\cite[Theorem XIII.64]{ReedSimonIV} we  obtain that $H_{\varphi{\sf E}}$ has a compact resolvent in $\K$. This means 
that formula $\eqref{wu2b}$ holds, the spectrum of  $H_{\varphi{\sf E}}$ is discrete\footnote{The spectrum of $H_{\varphi{\sf E}}$ is discrete and coincides with ${\sf E}'$ 
for $\lambda=\infty$. However, for finite $\lambda$, the spectrum of $H_{\varphi{\sf E}}$ includes both its discrete spectrum part ${\sf E}'$ 
and the point of essential spectrum $\lambda$.} and it coincides with ${\sf E}'$.
\end{proof}

\begin{rem}
The statement in Theorem \ref{kkk31} can be readily extended to cover the scenario of an increasing sequence $\ldots{\leq}E_j{\leq}{E_{j+1}}\leq\ldots$
 provided that eigenvalues have finite multiplicities.
 \end{rem}

Theorem \ref{kkk31} establishes a sufficient condition for the existence of the discrete spectrum of $H_{\varphi{\sf E}}$ 
without the need for explicit construction. The Min-Max principle \cite[p. 265]{Schmudgen} can be utilized to compute the eigenvalues. 
Additionally, an alternative method is presented below, which is specifically tailored to the properties of operators $H_{\varphi{\sf E}}$.

\begin{prop}\label{K10}
Let $H_{{\varphi}\sf E}$ be a Hamiltonian generated by the pair $(\F_\varphi, {\sf E})$  and let ${\mathcal R}(\theta_\varphi)$ be the image set of the analysis operator $\theta_\varphi$ associated with $\F_\varphi$ (see \eqref{K31}). Then $\mu\in\sigma_p(H_{{\varphi}\sf E})$ if and only if 
 there exists a sequence $\{c_j\}\in\mathcal{R}(\theta_\varphi)$  such that  
\begin{equation}\label{goa1}
 \{(E_j-\mu)c_j\}\in\ell_2(\mathbb{J})\ominus{\mathcal R}(\theta_\varphi)=\mathcal{R}(\theta_\psi).
\end{equation}   
The corresponding eigenvector of  $H_{{\varphi}\sf E}$ coincides with $f=\sum_{j\in\mathbb{J}}c_j \varphi_j$.  
\end{prop}
\begin{proof} 
Assume that    $H_{{\varphi}\sf E}f=\mu{f}$ for some $f\in\mathcal{D}(H_{{\varphi}\sf E})$. Since $\F_\varphi$ is a PF, $f=\sum_{j\in\mathbb{J}}\langle \varphi_j, f \rangle\varphi_j$ and the relation $H_{{\varphi}\sf E}f=\mu{f}$ takes the form
\begin{equation}\label{rrr1}
\sum_{j\in\mathbb{J}}(E_j-\mu)\langle \varphi_j, f \rangle\varphi_j=\sum_{j\in\mathbb{J}}(E_j-\mu)c_j\varphi_j=0, \qquad c_j=\langle \varphi_j, f \rangle.
\end{equation}

In view of \eqref{K31}, $\theta_\varphi{f}=\{c_j\}$, i.e., the sequence $\{c_j\}$
belongs to $\mathcal{R}(\theta_\varphi)$. Moreover, 
$\{E_jc_j\}\in\ell_2(\mathbb{J})$
since $f\in\mathcal{D}(H_{{\varphi}\sf E})$. Combining  \eqref{K13} with \eqref{rrr1} we
obtain \eqref{goa1}.

Conversely, if \eqref{goa1} holds, then $\{(E_j-\mu)c_j\}\in\ell_2(\mathbb{J})$ and
$\sum_{j\in\mathbb{J}}(E_j-\mu)c_j\varphi_j=0$, by virtue of \eqref{K13}. 
The sequence $\{c_j\}\in\mathcal{R}(\theta_\varphi)$ determines a vector 
$f=\sum_{j\in\mathbb{J}}c_j \varphi_j\in\K$, where $c_j=\langle \varphi_j, f \rangle$.
It follows from \eqref{goa1} that $\{E_j\langle \varphi_j, f\rangle\}\in\ell_2(\mathbb{J})$. This means that  $f\in\mathcal{D}(H_{{\varphi}\sf E})$ and the relation 
$\sum_{j\in\mathbb{J}}(E_j-\mu)c_j\varphi_j=\sum_{j\in\mathbb{J}}(E_j-\mu)\langle \varphi_j, f\rangle \varphi_j=0$
is equivalent to $H_{{\varphi}\sf E}f-\mu{f}=0$.
\end{proof}

 Typically, a PF $\F_\varphi=\{\varphi_j, j\in\mathbb{J}\}$ comprises a Riesz basis component $\F_\varphi^0=\{\varphi_j, j\in\mathbb{J}_0\}$ ($\mathbb{J}_0\subset\mathbb{J}$). 
 This is especially true when $\F_\varphi$ has finite excess.
In this case, the Hamiltonian $H_{{\varphi}\sf E}$ generated by the pair $(\F_\varphi, {\sf E})$ can be decomposed:
\begin{equation}\label{bbb22}
H_{{\varphi}\sf E}=A_0+A_1,  \qquad A_i=\sum_{j\in\mathbb{J}_i}E_j\langle \varphi_j, \cdot \rangle\varphi_j.
\end{equation} 
An additional analysis leads to the following: 
\begin{prop}\label{kkk21}
Assume that a PF $\F_\varphi=\{\varphi_j, j\in\mathbb{J}\}$
can be decomposed $\F_\varphi=\F_\varphi^0\cup\F_\varphi^1$
in such a way that $\F_\varphi^0=\{\varphi_j , j\in\mathbb{J}_0\}$ and $\F_\varphi^1=\{\varphi_j , j\in\mathbb{J}_1\}$ are, respectively, a Riesz basis and a frame sequence in $\K$. Denote by $S_0$ the frame operator for $\F_\varphi^0$ and suppose that $\sup_{j\in\mathbb{J}_1}\{|E_j|\}<\infty$. Then, for 
all $f\in\mathcal{D}(H_{{\varphi}\sf E})$,
\begin{equation}\label{jjj1}
H_{{\varphi}\sf E}f=S_0^{1/2}H_{e{\sf E}_0}S_0^{1/2}f+(I-S_0)^{1/2}H_{e{\sf E}_1}(I-S_0)^{1/2}f 
\end{equation}
where the self-adjoint operator  $H_{e{\sf E}_0}$ is defined by \eqref{bbb1}, the set
 $\{e_j,  j\in\mathbb{J}_1\}$ is  a PF of 
the subspace $\K_1=\overline{{\sf span}\ \F_\varphi^1}=\mathcal{R}(I-S_0)$, and 
 $H_{e{\sf E}_1}=\sum_{j\in\mathbb{J}_1}E_j \langle e_j, \cdot\rangle\,e_j$ is a bounded
 self-adjoint operator in $\K_1$.
\end{prop}
 \begin{proof} It follows from the proof of Corollary \ref{uman7} that the operator $A_0$ in \eqref{bbb22} coincides with
 $S_0^{1/2}H_{e{\sf E}_0}S_0^{1/2}$, where   $H_{e{\sf E}0}$ is defined by \eqref{bbb1}. 

 The operator $I-S_0$ is nonnegative in $\K$ since 
\begin{equation}\label{krakow1}
((I-S_0)f, f)=\|f\|^2-\sum_{j\in\mathbb{J}_0}|\langle \varphi_j, f \rangle)|^2=\sum_{j\in\mathbb{J}_1}|\langle \varphi_j, f \rangle|^2\geq{0}, \quad f\in\K.
\end{equation}
This means that the square root  $(I-S_0)^{1/2}$ exists. Further, the relation \eqref{krakow1}  implies that 
$$
\ker(I-S_0)=\K\ominus\K_1, \qquad \K_1=\overline{{\sf span}\ \F_\varphi^1}.
$$
Hence, $\K_1$ coincides with $\overline{\mathcal{R}(I-S_0)}$ and it is a reducing subspace for $I-S_0$. 
Denote by $(I-S_0)|_{\K_1}$ the restriction of $I-S_0$ onto $\K_1$. The relation
$$
(I-S_0)f=\sum_{j\in\mathbb{J}}\langle \varphi_j, f \rangle\varphi_j -\sum_{j\in\mathbb{J}_0}\langle \varphi_j, f\rangle\varphi_j=\sum_{j\in\mathbb{J}_1}\langle \varphi_j, f\rangle\varphi_j, \quad f\in\K_1
$$
implies that $(I-S_0)|_{\K_1}$
is  a frame operator of the frame $\F_\varphi^1$  in the Hilbert space $\K_1$. Hence, the inverse operator 
$((I-S_0)_{\K_1})^{-1} : \K_1 \to \K_1$ is bounded and $\overline{\mathcal{R}(I-S_0)}=
{\mathcal{R}(I-S_0)}$.

It follows from  \cite[Theorem 6.1.1]{chri} that the elements $\varphi_j\in\F_\varphi^1$ have the form
$$
  \varphi_j=(I-S_0)^{1/2}e_j, \qquad j\in\mathbb{J}_1,
  $$
  where $\{e_j, j\in\mathbb{J}_1\}$ is a PF  of $\K_1$. Moreover, repeating the proof of the relation \eqref{bbb1} for 
  the case of operator $A_1$ acting in $\K_1$ we get $A_1=(I-S_0)^{1/2}H_{e{\sf E}_1}(I-S_0)^{1/2}$.
Here, the operator $H_{e{\sf E}_1}$ is bounded self-adjoint in $\K_1$ due to the part (i) of Corollary \ref{uman7}. 
\end{proof}

Considering $H_{{\varphi}\sf E}$ as a perturbation of $A_0$ by  $A_1$ and supposing that $A_1$ is sufficiently small with respect to $A_0$
one can expect the coincidence of essential spectra of  $H_{{\varphi}\sf E}$ and $A_0$. 
For example, If $\F_\varphi$ has a finite excess, then $I-S_0$ is a compact operator, and therefore the second operator in \eqref{jjj1} is also compact.
In such a case, the classical Weyl theorem \cite[p. 182]{Schmudgen} implies that $\sigma_{ess}(H_{{\varphi}\sf E})=\sigma_{ess}(A_0+A_1)=\sigma_{ess}(A_0)$.

\section{Examples}\label{sect3}

This section is devoted to a detailed analysis of two examples of our previous results, with some preliminary applications to quantum mechanics. 

\subsection{Hamiltonians generated by the Casazza-Christensen frame}\label{s3.1}
In general, a PF with infinite excess may not contain a Riesz basis as a subset. If a subsequence of the frame elements is allowed to converge to 0 in norm, then it is easy to construct a frame that does not contain a Riesz basis. However, answering a similar question for frames that are norm-bounded below is much more complicated. An example of such a PF that is norm-bounded below, but does not contain a Schauder basis, was first constructed by Casazza and Christensen \cite{CC98}.
 Our aim now is to investigate Hamiltonians generated by that PF and discuss their possible physical applications.

Let  $\K_n$ ($n\in\mathbb{N}$) be a $n$-dimensional Hilbert space with ONB $\{{e}_1^{(n)}, {e}_2^{(n)}, \ldots {e}_n^{(n)}\}$. 
Then the   set 
$\F_\varphi^{(n)}=\{\varphi_j^{(n)}, j\in\mathbb{J}_n\}$ where $\mathbb{J}_n=\{1, 2, \ldots, n, n+1\}$ and
\begin{equation}\label{K61}
\varphi_j^{(n)}={e}_j^{(n)}-\frac{1}{n}\sum_{i=1}^{n}{e}_i^{(n)},  \quad 1\leq{j}\leq{n}, \qquad \varphi_{n+1}^{(n)}=\frac{1}{\sqrt{n}}\sum_{i=1}^{n}{e}_i^{(n)},
\end{equation}
is a PF for $\K_n$ \cite[Lemma 7.5.1]{chri}.  The operator 
\begin{equation}\label{palermo40}
{H_{{\varphi_n}{\sf E}_n}}=\sum_{j\in\mathbb{J}_n}E_j^{(n)} \langle \varphi_j^{(n)}, \cdot\rangle\,\varphi_j^{(n)}=\sum_{j=1}^{n+1}E_j^{(n)} \langle \varphi_j^{(n)}, \cdot\rangle\,\varphi_j^{(n)}, \qquad {\sf E}_n=\{E_j^{(n)}, j\in\mathbb{J}_n\}.
\end{equation}
is a bounded self-adjoint operator in $\K_n.$

 For the PF $\F_\varphi^{(n)}$, the complementary Hilbert space $\M_n$ can be chosen as $\mathbb{C}$ 
and the complementary PF $\F_\psi^{(n)}=\{\psi_j^{(n)}, j\in\mathbb{J}_n\}$ in \eqref{K3b} has the form \cite{onDF}:
\begin{equation}\label{palermo25}
\psi_j^{(n)}=\frac{1}{\sqrt{n}}, \quad j=1,\ldots{n}, \qquad \psi_{n+1}^{(n)}=0.
\end{equation}

Consider the direct sum 
$\K=\left(\sum_{n=1}^\infty\oplus\K_n\right)_{\ell_2}.$ 
The Hilbert space $\K$ consists of sequences $f=(f_1, f_2, \ldots)$ for which\footnote{we identify elements of $\K_n$ with their counterpart in  $\K$, i.e., 
we do not distinguish between $f\in\K_n$ and the sequence in $\K$ having $f$ in the $n$-th entry and otherwise zero.} $f_n\in\K_n$ and 
$\sum_{n=1}^\infty\|f_n\|_{\K_n}^2<\infty$.
The scalar product in $\K$ is defined as follows 
$$ 
\langle f, g\rangle=\sum_{n=1}^\infty\langle f_n, g_n \rangle_{\K_n}, \qquad f, g \in \K. 
$$
Since $\K=\left(\sum_{n=1}^\infty\oplus\K_n\right)_{\ell_2}$, the union of PF's 
$\F_{\varphi}=\bigcup_{n=1}^\infty\F_{\varphi}^{(n)}$ is a PF for $\K$ \cite[Theorem 7.5.2]{chri}.    
The complementary Hilbert space $\M$ in Theorem \ref{Naimark} for the PF $\F_\varphi$ can be chosen as the direct sum
$$
\M=\left(\sum_{n=1}^\infty\oplus\M_n\right)_{\ell_2}=\left(\sum_{n=1}^\infty\oplus\mathbb{C}\right)_{\ell_2}=\ell_2(\mathbb{N}).
$$ 
The  PF $\F_\psi$ in $\M$ coincides with  the union of  PF's: $\F_{\psi}=\bigcup_{n=1}^\infty\F_{\psi}^{(n)}$.

Denote 
$$
{\sf E}=\bigcup_{n=1}^\infty{\sf E}_n=\{E_j^{(n)}, \ n\in\mathbb{N}, \ 1\leq{j}\leq{n+1}\}.
$$

For the PF $\F_\varphi$ and ${\sf E}$ we consider the operator $H_{{\varphi}\sf E}^{min} : \K \to \K$  defined by  \eqref{K22}.
$$
H_{{\varphi}\sf E}^{min}=\sum_{n=1}^\infty\sum_{j=1}^{n+1}E_j^{(n)}\langle \varphi_j^{(n)}, \cdot\rangle\,\varphi_j^{(n)}, \quad \mathcal{D}(H_{{\varphi}\sf E}^{min})=\{f\in\K \ : \ \sum_{n=1}^\infty\sum_{j=1}^{n+1}(E_j^{(n)})^2 |\langle \varphi_j^{(n)},  f\rangle|^2<\infty\}.
$$

Taking into account the decomposition $\K=\left(\sum_{n=1}^\infty\oplus\K_n\right)_{\ell_2}$  and \eqref{palermo40} one can rewrite 
the last formulas as follows
$$
H_{{\varphi}\sf E}^{min}=\sum_{n=1}^\infty\oplus{H_{{\varphi_n}{\sf E}_n}}, \qquad \mathcal{D}({H_{{\varphi_n}{\sf E}_n}^{min}})=\{f\in\K \ : \ \sum_{n=1}^\infty\|H_{{\varphi_n}{\sf E}_n}\|_{\K_n}^2<\infty\}.
$$
The obtained formula means that $H_{{\varphi}\sf E}^{min}$ is self-adjoint in $\K$ 
(since $H_{{\varphi_n}{\sf E}_n}$  are bounded self-adjoint operators in $\K_n$ for all $n\in\mathbb{N}$).
 By Corollary \ref{dd1},  $H_{{\varphi}\sf E}:=H_{{\varphi}\sf E}^{min}=H_{{\varphi}\sf E}^{max}$ 
 is a Hamiltonian generated by the pair $(\F_\varphi, {\sf E})$. 

As the PF $\F_\varphi$ does not contain a Riesz basis \cite{CC98}, Proposition \ref{kkk21} cannot be applied for the investigation of 
$H_{{\varphi}\sf E}$. Nonetheless, one can attempt to use Theorem  \ref{kkk31} and Proposition \ref{K10}.

Assume that $\{E_j^{(n)}\}$ is a strictly increasing sequence, i.e. 
$$
E_1^{(1)}<E_2^{(1)}<E_1^{(2)}<E_{2}^{(2)}<E_3^{(2)}\ldots<E_1^{(n)}<E_2^{(n)}\ldots<E_n^{(n)}<E_{n+1}^{(n)}<\ldots,
$$
where $E_j^{(n)} \to \infty$.  Taking \eqref{palermo25} into account one gets that the operator $B : \K \to \ell_2(\mathbb{N})$ defined by  \eqref{palermo33}
has the form
\begin{equation}\label{bbb46} 
Bf=\left\{\frac{1}{\sqrt{n}}\sum_{j=1}^{n}{E_j^{(n)}} \langle \varphi_j^{(n)}, f\rangle\right\}_{n=1}^\infty, \qquad f\in\mathcal{D}(B)\subseteq\K.
\end{equation}
 If 
$$ 
\sup_{n\in\mathbb{N}, \ 1\leq{j}\leq{n}}\frac{1}{\sqrt{n}}E_j^{(n)}<\infty,
$$
then the right-hand side of \eqref{bbb46} belongs to $\ell_2(\mathbb{N})$ for every $f\in\K$ and the operator $B$ is bounded.
In this case, applying Theorem \ref{kkk31} we arrive at the conclusion that $H_{{\varphi}\sf E}$ has a discrete spectrum.
Explicit calculation of the discrete spectrum can be carried out using Proposition \ref{K10}.  
Building on the argumentation presented in \cite{BK}, we conclude that the discrete spectrum of $H_{\varphi{\sf E}}$
consists of the original quantities\footnote{this fact immediately follows from Corollary \ref{palermo1} since $\psi_{n+1}^{(n)}=0$, see \eqref{palermo25}} $\{E_{n+1}^{(n)}\}_{n=1}^\infty$  
and the solutions $\mu_1, \ldots \mu_{n-1}$ of the equations 
$$
\frac{1}{E_1^{(n)}-\mu}+\frac{1}{E_2^{(n)}-\mu}+\ldots+\frac{1}{E_{n}^{(n)}-\mu}=0, \qquad n\geq{2}.
$$
The corresponding eigenfunctions are 
$$
\varphi_{n+1}^{(n)}=\frac{1}{n}\sum_{i=1}^{n}{e_i^{(n)}} \quad  (\mbox{for the eigenvalues} \quad E_{n+1}^{(n)})  \quad \mbox{and} \quad
f_j=\sum_{i=1}^{n}\frac{1}{E_i^{(n)}-\mu_j}{e}_i^{(n)},  
$$ 
for the eigenvalues $\mu_j$, $1\leq{j}\leq{n-1}$ ($n\geq{2}$).

This example can be used to produce a natural settings in the realm of signal analysis. For that, we need first to introduce some ladder operators. 
In particular, we will now introduce the {\em horizontal} ladder operators $a_n$ and $a_n^*$, acting on $\K_n$, and the {\em vertical} ladder operators $V_n$ and $V_n^*$, mapping $\K_{n+1}$ into $\K_{n}$ and vice-versa. 

First of all we define
$$
a_ne_j^{(n)}=\left\{ \begin{array}{c}
	0, \qquad \,\,\, \mbox{ if } j=1 \\
	\sqrt{j-1}\,e^{(n)}_{j-1}, \qquad  \mbox{ if } j=2,3,\ldots,n,
\end{array}\right.
$$
whose adjoint is 
$$
a_n^* e_j^{(n)}=\left\{ \begin{array}{c}
	\sqrt{j}\,e^{(n)}_{j+1}, \qquad\qquad  \mbox{ if } j=1,2,\ldots,n-1,\\
	0, \qquad  \mbox{ if } j=n. 
\end{array}\right.
$$
These operators, already introduced in \cite{baggarg} in connection with a biological system, satisfy the following commutation rule:
$$
[a_n,a_n^*]=I_{\K_n}-nP_n^{(n)},
$$
where $P_n^{(n)}f=\langle e_n^{(n)},f\rangle_{\K_n}e_n^{(n)}$, for all $f\in\K_n$,  and where $I_{\K_n}$ is the identity operator on $\K_n$.

The action of $a_n$ on the vectors $\varphi_j^{(n)}$ defined by \eqref{K61} is as follows:
$$
a_n\varphi_j^{(n)}=\left\{ \begin{array}{c}
	-\tilde e^{(n)}, \qquad\qquad\qquad\qquad \,\,\, \mbox{ if } j=1 \\
	\sqrt{j-1}\,\varphi^{(n)}_{j-1}+ \frac{\sqrt{j-1}}{\sqrt{n}}\,\varphi_{n+1}^{(n)}	
	-\tilde e^{(n)}, \qquad  \mbox{ if } j=2,3,\ldots,n,\\
	\sqrt{n}\, \tilde e^{(n)}, \qquad\qquad\qquad\qquad \,\,\, \mbox{ if } j=n+1,
\end{array}\right.
$$
 where  
$$
\tilde e^{(n)}=\frac{1}{n}\sum_{i=1}^{n}\sqrt{i}\,{e}_i^{(n)}.
$$
It is easy (but not so relevant) to deduce also the action of $a_n^*$ on $\varphi_j^{(n)}$. 

We can further introduce the operator $V_{n+1}:\K_{n+1}\rightarrow\K_n$ as follows:  
$$
V_{n+1}f=\sum_{j=1}^{n+1}\langle e_j^{(n+1)},f\rangle_{\K_{n+1}}\,\varphi_j^{(n)}, \qquad f\in\K_{n+1}.
$$
In particular, we see that $V_{n+1}e_j^{(n+1)}=\varphi_j^{(n)}$. The adjoint of $V_{n+1}$ can be easily deduced. It is an operator mapping $\K_n$ into $\K_{n+1}$ as follows:
$$
V_{n+1}^*g=\sum_{j=1}^{n+1}\langle \varphi_j^{(n)},g\rangle_{\K_{n}}\,e_j^{(n+1)}, \qquad g\in\K_n.
$$ 
In this case, it is clear that $V_{n+1}^*\varphi_j^{(n)}\neq e_j^{(n+1)}$, in general. However, due to the fact that $\{\varphi_j^{(n)}\}_{j=1}^{n+1}$ is a PF in $\K_n$, 
it is possible to check that $V_{n+1}V_{n+1}^*=I_{\K_n}$, while $V_{n+1}^*V_{n+1}\neq{I}_{\K_{n+1}}$.

As for the possible interpretation of this example, and the ladder-like operators $a_n$, $V_{n+1}$ and their adjoints, 
a natural look at this framework is in terms of signal analysis: $\K_n$ is the set of signals with $n$ bits. If a signal $f$ is $f=e_1^{(n)}$, 
this means that only the first bit is "on", in a signal of $n$ bits. Analogously, if $f=\alpha_1e_1^{(n)}+\alpha_ne_n^{(n)}$, then the first and the last bits are "on", 
with two weights $\alpha_1$ and $\alpha_n$, with $|\alpha_1|^2+|\alpha_n|^2=1$. 
The operators $a_n$ and $a_n^*$ switch on and off the various bits of a signal with $n$ bits, while $V_{n+1}$ and $V_{n+1}^*$ add or remove bits from the signal. 
Also, going from $\{e_j^{(e)}\}$ to $\{\varphi_j^{(e)}\}$ in this context is rather natural: a frame is what is often used in signal analysis to take into account possible loss 
of information during the transmission of the signal. The Hamiltonian can be seen as the energy of the signal, with various contributions arising from different possible
 lengths of the signals.
 
\subsection{Relations with regular pseudo-bosons}\label{s3.2}

This section is focused on a class of examples of PFs connected to the so-called {\em regular pseudo-bosons} \cite{baginbagbook,bagspringer}, which are suitable deformations of the bosonic ladder operators $c=\frac{1}{\sqrt{2}}\left(x+\frac{d}{dx}\right)$ and $c^*=\frac{1}{\sqrt{2}}\left(x-\frac{d}{dx}\right)$.

The starting point here is a bounded operator $X$ with bounded inverse $X^{-1}$  and an ONB $\F_e=\{e_n\in\K, \, n\in \mathbb{J}=\mathbb{N}\cup\{0\}\}$ of a Hilbert space $\K$. 
Then we have 
\begin{prop}\label{prop15}
	If $\|X^*X\|<1$, then the sets 
	$$\F_{\varphi}=\left\{\varphi_n=X^*e_n, \,n\in \mathbb{J}\right\}, \qquad  \F_{\tilde\varphi}=\left\{\tilde\varphi_n=(I-X^*X)^{1/2}e_n, \,n\in \mathbb{J}\right\}
	$$ are Riesz bases of $\K$, 
	with dual bases 
	$$\F_{\psi}=\left\{\psi_n=X^{-1}e_n, \,n\in \mathbb{J}\right\}, \qquad \F_{\tilde\psi}=\left\{\tilde\psi_n=(I-X^*X)^{-\,1/2}e_n, \,n\in \mathbb{J}\right\}.$$ 
	Moreover, the set $\F_{ex}^\varphi=\F_{\varphi}\cup\F_{\tilde\varphi}$  is a  PF of $\K$.
	
	If $\|X^*X\|>1$ then the set  $\F_{\widetilde{\widetilde\psi}}=\left\{\widetilde{\widetilde\psi_n}=(I-(X^*X)^{-1})^{1/2}e_n, \,n\in \mathbb{J}\right\}$
	is a Riesz basis with dual 
	$$
	\F_{\widetilde{\widetilde\varphi}}=\left\{\widetilde{\widetilde\varphi_n}=(I-(X^*X)^{-1})^{-\,1/2}e_n, \,n\in \mathbb{J}\right\}
	$$
	and the set  $\F_{ex}^\psi=\F_{\psi}\cup\F_{\widetilde{\widetilde\psi}}$ is a PF of $\K$. 
	\end{prop}
\begin{proof}
	If $\|X^*X\|<1$, then $I-X^*X$ is a bounded positive operator on $\K$,  and there exist a positive square root 
	$(I-X^*X)^{1/2}$ and its inverse $(I-X^*X)^{-1/2}$ that are bounded operators  on $\K$.  
	Hence the Riesz basis nature of the pairs $\F_{\varphi}$,  $\F_{\psi}$  and  $\F_{\tilde\varphi}$, $\F_{\tilde\psi}$ is clear. 
	To prove that $\F_{ex}^\varphi$ is a PF, let us put $U=X^*X$, which is bounded with bounded inverse and self adjoint, and let us write a generic $f\in\K$ as 
	$f=Uf+(I-U)f$. We have
	$$
	\langle Uf,f\rangle=\langle Xf,Xf\rangle=\sum_{n\in\mathbb{J}}|\langle Xf,e_n\rangle|^2=\sum_{n\in\mathbb{J}}|\langle f,\varphi_n\rangle|^2.
	$$
	Moreover, since $I-U$ is a positive operator, 
	$$
	(I-U)f=(I-U)^{1/2}\left((I-U)^{1/2}f\right)=(I-U)^{1/2}\sum_{n\in\mathbb{J}}\langle e_n,(I-U)^{1/2}f\rangle e_n=\sum_{n\in\mathbb{J}}\langle \tilde\varphi_n,f\rangle \tilde\varphi_n,
	$$
	so that $\langle (I-U)f,f\rangle=\sum_{n\in\mathbb{J}}|\langle \tilde\varphi_n, f\rangle|^2$. Hence
	$$
	\|f\|^2=\langle Uf,f\rangle+\langle (I-U)f,f\rangle=\sum_{n\in\mathbb{J}}|\langle f,\varphi_n\rangle|^2+\sum_{n\in\mathbb{J}}|\langle \tilde\varphi_n,f\rangle|^2,
	$$
	which implies our claim: $\F_{ex}^\varphi$ is a PF.
	
	The proof of the second part of the proposition is similar to this, and will not be repeated.
	
\end{proof}

\vspace{2mm}

{\bf Remarks:--} (1) It is clear that if the sets $\F_{\tilde\varphi}$ and $\F_{\tilde\psi}$ are well defined, then the other sets, 
$\F_{\widetilde{\widetilde\varphi}}$ and $\F_{\widetilde{\widetilde\psi}}$, are not. 
The point is that only one between $I-U$ and $I-U^{-1}$ can be positive. 
This does not exclude the possibility to extend the above construction to consider both these possibilities. But we will not investigate further this point here.

\subsubsection{An explicit construction}

As it is discussed in \cite{baginbagbook,bagspringer}, in particular, pseudo-bosons can be seen as suitable deformations of the ladder operators $c$ and $c^*$ 
we introduced before. It is well known that these operators can be used to diagonalize the Hamiltonian of a quantum harmonic oscillator $H_0=\frac{1}{2}(p^2+x^2)$, 
where $p$ and $x$ are respectively the momentum and the position operators (both self-adjoint). 
In fact, after some algebra we can write $H_0=c^*c+\frac{1}{2}\,I$, and its eigenstates $e_n$ can be constructed by fixing first the vacuum $e_0$, i.e. a vector in $\K$ such that $ce_0=0$, and then acting on it with powers of $c^*$: $e_n=\frac{(c^*)^n}{\sqrt{n!}}\,e_0$. Then, $\F_e=\{e_n, \,n\geq0\}$ is an ONB of $\K$. In particular, in the position representation, where $c$ and $c^*$ are the differential operators already introduced, we have $\K=\ltwo$ and
$$
e_n(x)=\frac{1}{\sqrt{2^nn!\sqrt{\pi}}}H_n(x)\,e^{-\frac{x^2}{2}},
$$
where $H_n(x)$ is the $n$-th Hermite polynomial. 

\vspace{2mm}

{\bf Remark:--} Despite the apparent simplicity of the system (the well known harmonic oscillator), it is maybe useful to stress that this is not necessarily trivial. 
Indeed, we could think of $H_0$ as the single electron Hamiltonian of a two-dimensional gas of electrons in a strong magnetic field, orthogonal to a fixed plane, 
when expressed in suitable variables. This is the physical system which is behind the Landau levels, and the fractional quantum Hall effect, see \cite{chakra,macdonald}
 for instance, and has attracted a lot of interest in the past years, \cite{fqhe4, fqhe1,fqhe2,fqhe3} among the others.  In most of these papers, the possibility of modifying a single wave function in the so-called {\em lowest Landau level} (LLL) is 
used to construct the wave function for the gas of electrons, localized at different lattice sites and minimizing the energy of the gas.

\vspace{2mm}

In what we will do now, we are inspired by the possibility of translating a wave function of the LLL while staying in the same energetic level. However, rather than using only translation operators, we will consider a combination of multiplication and translation operators. More explicitly, let $K$ and $T$ be the operators defined as follows:
$$
Kf(x)=m(x)f(x), \qquad Tf(x)=f(x-\alpha), \qquad  f\in\ltwo.
$$
Here $\alpha>0$, fixed, and $m(x)$ is a complex-valued, smooth\footnote{We could take, for instance, $m(x)\in C^\infty$.} function satisfying $0<m\leq|m(x)|\leq M<1$  in $\mathbb{R}$. These operators are bounded, with bounded inverse. In particular, $T$ is unitary:
$$
K^{-1}f(x)=\frac{1}{m(x)}\,f(x), \qquad K^*f(x)=\overline{m(x)}\,f(x), \qquad T^*f(x)=T^{-1}f(x)=f(x+\alpha).
$$
Now, if we put $X^*=TK$, we get the following results:
$$
X^*f(x)=m(x-\alpha)\,f(x-\alpha), \qquad Xf(x)=\overline{m(x)}f(x+\alpha),
$$
together with
$$
(X^{-1})^*f(x)=\frac{1}{m(x)}\,f(x+\alpha), \qquad  X^{-1}f(x)=\frac{1}{\overline{m(x-\alpha)}}\,f(x-\alpha).
$$
The operators $X^*X$ and $XX^*$ turn out to be both multiplication operators:
$$
X^*Xf(x)=|m(x-\alpha)|^2f(x), \qquad XX^*f(x)=|m(x)|^2f(x).
$$
Because of our assumption on $m$, we have $|m(x-\alpha)|^2\leq{M^2}<1$, and therefore $\|X^*X\|<1$: we are in the first case of Proposition \ref{prop15}, 
so that the sets $\F_{\widetilde{\widetilde\varphi}}$ and $\F_{\widetilde{\widetilde\psi}}$ cannot be defined. Still we find
$$
\varphi_n(x)=m(x-\alpha)e_n(x-\alpha), \qquad \psi_n(x)=\frac{1}{\overline{m(x-\alpha)}}e_n(x-\alpha)
$$
and
$$
\tilde\varphi_n(x)=\sqrt{1-|m(x-\alpha)|^2}\,e_n(x), \qquad \tilde\psi_n(x)=\frac{1}{\sqrt{1-|m(x-\alpha)|^2}}\,e_n(x).
$$
Biorthonormality (in pairs) of these functions is manifest, while the fact that $\F_{ex}^\varphi=\F_{\varphi}\cup\F_{\tilde\varphi}$ is a PF is not as clear, but it is a consequence of Proposition \ref{prop15}.

Following \cite{bagspringer} it is easy to find the ladder operators for $\F_\varphi$ and $\F_{\tilde\varphi}$, and for their dual Riesz bases. We introduce the operators $a_\varphi$, $b_\varphi$, $a_{\tilde\varphi}$ and $b_{\tilde\varphi}$ as follows:
$$
a_\varphi f(x)=X^*c(X^*)^{-1}f(x), \qquad b_\varphi f(x)=X^*c^*(X^*)^{-1}f(x),
$$
and
$$
a_{\tilde\varphi} f(x)=(I-X^*X)^{1/2}c(I-X^*X)^{-1/2}f(x), \qquad b_{\tilde\varphi} f(x)=(I-X^*X)^{1/2}c^*(I-X^*X)^{-1/2}f(x),
$$
$\forall f(x)\in\scr$, the set of the $C^\infty$, fast decreasing, functions (the Schwartz space). Simple computations allow us to deduce the following expressions:
$$
a_\varphi=c-\frac{1}{\sqrt{2}}\left(\alpha+\frac{m'(x-\alpha)}{m(x-\alpha)}\right), \qquad b_\varphi=c^*-\frac{1}{\sqrt{2}}\left(\alpha-\frac{m'(x-\alpha)}{m(x-\alpha)}\right),
$$
while
$$
a_{\tilde\varphi}=c-\frac{1}{\sqrt{2}}\frac{q'(x)}{q(x)}, \qquad b_{\tilde\varphi}=c^*+\frac{1}{\sqrt{2}}\frac{q'(x)}{q(x)},
$$
where $q(x)=\sqrt{1-|m(x-\alpha)|^2}$. Incidentally we observe that we can rewrite $\frac{q'(x)}{q(x)}=\frac{d(\log(q(x)))}{dx}$.

Pseudo-bosonic operators are useful since they act as ladder operators on the families of function deduced before. In particular, we have
$$
a_\varphi\varphi_n(x)=\sqrt{n}\varphi_{n-1}(x), \qquad b_\varphi\varphi_n(x)=\sqrt{n+1}\varphi_{n+1}(x), 
$$
and similarly
$$
a_{\tilde\varphi}\tilde\varphi_n(x)=\sqrt{n}\tilde\varphi_{n-1}(x), \qquad b_{\tilde\varphi}\tilde\varphi_n(x)=\sqrt{n+1}\tilde\varphi_{n+1}(x), 
$$
with the understanding that $\varphi_{-1}(x)=\tilde\varphi_{-1}(x)=0$. As a consequence, the various $\varphi_n(x)$ are eigenstates of $N_\varphi=b_\varphi a_\varphi$, while  each $\tilde\varphi_n(x)$ is an eigenstate of $N_{\tilde\varphi}=b_{\tilde\varphi} a_{\tilde\varphi}$, both with eigenvalue $n$. If we further compute the adjoints of these operators, then we obtain ladder operators 
for the dual families, $\F_\psi$ and $\F_{\tilde\psi}$, \cite{baginbagbook, bagspringer}.

Defining now the vectors $\Phi_n$ as
$$
\Phi_n(x)=\left\{\begin{array}{c} 
	\varphi_n(x), \qquad\,\,\,\,\,\, n=0,1,2,3,\ldots \\
	\tilde{\varphi}_{-n-1}(x), \qquad\,\,\, {n=-1,- 2, - 3,\ldots}
\end{array}\right.
$$
and the set $\F_\Phi=\{\Phi_n(x),\,n\in\mathbb{Z}\}{=\F_\varphi\cup\F_{\tilde{\varphi}}}$, it is easy to check that, taken any $f(x)\in\ltwo$,
$$
\sum_{n\in\mathbb{Z}}\langle\Phi_n,f\rangle\Phi_n(x)=\sum_{n=0}^\infty\langle\varphi_n,f\rangle\varphi_n(x)+\sum_{n=0}^\infty\langle\tilde{\varphi}_n,f\rangle\tilde{\varphi}_n(x)=X^*Xf(x)+(1-X^*X)f(x)=f(x),
$$
as it should {(since $\F_\varphi\cup\F_{\tilde{\varphi}}$ is a PF in $\K=\ltwo$)}. Now, given a set of (real) numbers $E_n$, $n\in\mathbb{Z}$, we can consider an operator
$$
H=\sum_{n\in\mathbb{Z}}E_n\langle\Phi_n,\cdot\rangle\Phi_n(x)=H_1+H_2,
$$
where 
$$
H_1=\sum_{n=0}^\infty E_n\langle\varphi_n,\cdot\rangle\varphi_n(x), \qquad H_2=\sum_{n=0}^\infty E_{-(n+1)}\langle\tilde\varphi_n,\cdot\rangle\tilde\varphi_n(x).
$$
If we now fix $E_n=E_{-(n+1)}=n$, then 
$$
H_1=\sum_{n=0}^\infty n\langle\varphi_n,\cdot\rangle\varphi_n(x)=N_\varphi \sum_{n=0}^\infty \langle\varphi_n,\cdot\rangle\varphi_n(x), 
$$ while $H_2=\sum_{n=0}^\infty n\langle\tilde\varphi_n,\cdot\rangle\tilde\varphi_n(x)=N_{\tilde\varphi} \sum_{n=0}^\infty \langle\tilde\varphi_n,\cdot\rangle\tilde\varphi_n(x)$.

As for a possible interpretation of these two terms, we can go back to our previous remark, and to the explicit expressions of the functions $\varphi_n(x)$ and $\tilde\varphi_n(x)$. In particular, while these latter are proportional (via the weight function $\sqrt{1-|m(x-\alpha)|^2}$) to the $e_n(x)$, the functions $\varphi_n(x)$ are again proportional  (but via the other weight function $m(x-\alpha)$) to the translated version of the $e_n(x)$. Then, while $H_2$ can be seen as the single-electron deformed Hamiltonian for a particle in the lowest Landau level, $H_1$ can be seen as its shifted version (with a different weight function). This can be interesting in connection with the crystals constructed out of the single electron, as in 
\cite{fqhe4, fqhe1, fqhe2, fqhe3}, since it could produce two different lattices, one shifted with respect to the other. In condensed matter,  lattices of this kind are useful, like the so-called reciprocal lattices.

\vspace{2mm}

These examples do not cover all possible applications of the general strategy proposed in this paper. More results, and more applications, are part of our future plans.

\section*{Acknowledgments}

FB was partially supported by the University of Palermo, via the CORI 2017 action, and by the Gruppo Nazionale di Fisica Matematica of Indam.
SK was partially supported by Faculty of Applied Mathematics AGH UST statutory tasks within subsidy of Ministry of Science and Higher Education, and
by the GNFM. SK expresses their gratitude to the University of Palermo for its warm hospitality.

\end{document}